\documentclass[12pt]{article}
\usepackage{multirow}
\usepackage[a4paper, total={7in, 10in}]{geometry}
\usepackage{titling}
\usepackage{natbib}
\usepackage[utf8]{inputenc}
\usepackage{booktabs}
\usepackage{float}
\usepackage{placeins}
\usepackage{amsmath}
\usepackage{amssymb}
\usepackage{doc}
\usepackage{url}
\usepackage{adjustbox}
\usepackage{tabularx}
\usepackage{graphicx}
\usepackage{epstopdf}
\usepackage{amsthm}
\newtheorem*{theorem}{Theorem}
\title{\textbf{On an Induced Distribution and its Statistical Properties}}

\author{Brijesh P. Singh \& Utpal Dhar Das \\\\
	Department of Statistics, Institute of Science\\
	Banaras Hindu University, Varanasi-221105}
\date{}

\begin{document}
\maketitle

\abstract
 In this study an attempt has been made to propose a way to develop new distribution. For this purpose, we need only idea about distribution function. Some important statistical properties of the new distribution like moments, cumulants, hazard and survival function has been derived. The renyi entropy, shannon entropy has been obtained. Also ML estimate of parameter of the distribution is obtained, that is not closed form. Therefore, numerical technique is used to estimate the parameter. Some real data sets are used to check the suitability of this distribution over some other existing distributions such as Lindley, Garima, Shanker and many more. AIC, BIC, -2loglikihood, K-S test suggest the proposed distribution works better than others distributions considered in this study.\\
 

\section{Introduction}
Almost all applied sciences including, biomedical science, engineering, finance, demography, environmental and agricultural sciences, there is a need of statistical analysis and modeling of the data. A number of continuous distributions for modeling lifetime data have been introduced in statistical literature such as Exponential, Lindley, Gamma, Lognormal and Weibull. Among these Gamma and Lognormal distributions are less popular because their survival functions cannot be expressed in closed forms and both require numerical integration. Researchers in probability distribution theory often use a probability distributions based on either their mathematical simplicity or because of their flexibility. Several parametric models are used in the analysis of lifetime data and in the problems associated with the modeling of the failure process. The Exponential distribution is often used to model the time interval between successive random events but Gamma and Weibull distribution is the most widely used model for lifetime distribution due to its flexibility. The exponential distribution is a particular case of the Gamma and Weibull distribution. In order to increase the suitability of the well-known distributions, many authors have proposed different transformations to generate new distributions, it has been an increased interest in defining new generators for univariate continuous distributions by introducing one or more additional shape parameter(s) to the baseline model. This improvs the goodness-of-fit of the proposed generated distribution.\\

\noindent In the context of increasing flexibility in distribution, many generalization or transformation methods are available in the literature
based on baseline distribution. \cite{ghitany2011two} developed a two-parameter weighted Lindley distribution and discussed its applications to survival data. \cite{zakerzadeh2009generalized} obtained a generalized Lindley distribution and discussed its various properties and applications. \cite{shaw2007alchemy} propose a new transformation method by adding one extra parameter. \cite{kumaraswamy1980generalized} gives another method of proposing new distribution by taking baseline distribution. \cite{abd2013utilizing} introduce Bilal distribution as a member of the families of distributions for the median of a random sample drawn from an arbitrary lifetime distribution. Since its failure rate function is monotonically increasing with finite limit for this they generalize Bilal distribution by making transformation $X=\left(\frac{Y-\delta}{\theta}\right)^{\lambda}$, the parameter $\delta$ is a threshold parameter, $\theta$ and $\lambda$ are the scale and the shape
parameters, respectively. \cite{gupta1998modeling} propose exponentiated type distribution by adding one more shape parameter. \cite{cordeiro2013beta} proposed a new class of distribution by adding two additional shape parameters. Also some well-known generators are the beta-G by \cite{eugene2002beta}, gamma-G by \cite{zografos2009families}, the Zografos-Balakrishnan-G family by \cite{nadarajah2015zografos}.\\

\section{Genesis of the distribution}
\noindent In this paper, an attempt has been made to develop a new distribution using concept of induced distribution. Let $X$ be a continuous random variable with probability density function $f(x)$, the cumulative distribution function $F(x)$ and expectation $E(x)$, we know that 
\begin{align}
	E(x)=\int\limits_{0}^{\infty}[1-F(x)]dx\qquad<\infty\nonumber
\end{align}
Now we define a p.d.f $g^{*}(x)$ as
\begin{align}
	g^{*}(x)=\frac{1-F(x)}{E(x)};\qquad x>{0}
\end{align}
Then $g^{*}(x)$ is called an induced distribution \cite{gupta1990role}. Actually this distribution is a particular case of weighted distribution defined by \cite{patil1977weighted}. According to the \cite{patil1977weighted}, if $f(x;\theta)$  be the probability distribution function of random variable \textit{X}  and the unknown parameter $\theta$  the weighted distribution is defined as;

\begin{center}
	$f(x;\theta)=\dfrac{w(x)f^*(x;\theta)}{E[w(x)]}\quad;\quad x\in \mathbb{R},\theta>0 $
\end{center}
where $w(x)$  is the weight function,  and $f^*(x;\theta)$  is the base line distribution. We know that $1-F(x)=s(x)=\frac{f(x)}{h(x)}$, i.e if we take $w(x)={h(x)^{-1}}$, we can get the induced distribution defined above in equation number (1). This distribution is well connected to its parent distribution and many of the statistical properties can be easily studied.

\section{Proposed Distribution}
We consider cdf of Lindley distribution and using the idea of induced distribution given in the equation (1), the pdf and cdf given in equation (2) and (3) respectively have been obtained. This distribution is already introduced by \cite{Garima2018Shanker} and known as Garima distribution, which is a mixture of Exponential ($\theta$) and Gamma (2,$\theta$) distribution with mixing proportion $\frac{\theta+1}{\theta+2}$. 
\begin{align}
		f(x;\theta)&=\frac{\theta}{\theta+2}\left(1+\theta+\theta x\right)e^{-\theta x}\\
		F(x;\theta)&=1-\left[1+\frac{\theta x}{\theta+2}\right]e^{-\theta x};x>0,\theta>0.
\end{align}
 Now, we consider cdf $F(x)$ of Garima distribution as a base line distribution and try to develop new distribution. The pdf and cdf of the new distribution is as follows
\begin{align}
	g(x;\theta)=\frac{\theta}{\theta+3}\left(2+\theta+\theta x\right)e^{-\theta x};\qquad x>0,\theta>0.
\end{align}
and the corresponding cdf is 
\begin{align}
	G(x;\theta)&=1-\left[1+\frac{\theta x}{\theta+3}\right]e^{-\theta x};\qquad x>0,\theta>0.
\end{align}
The above distribution is similar to the base line distribution and develop using concept of induced distribution thus named as induced Garima ($i$-Garima) distribution. This distribution can also be consider as second order induced Lindley distribution.
\begin{figure}[H]
	\centering
	\includegraphics[width=0.75\linewidth]{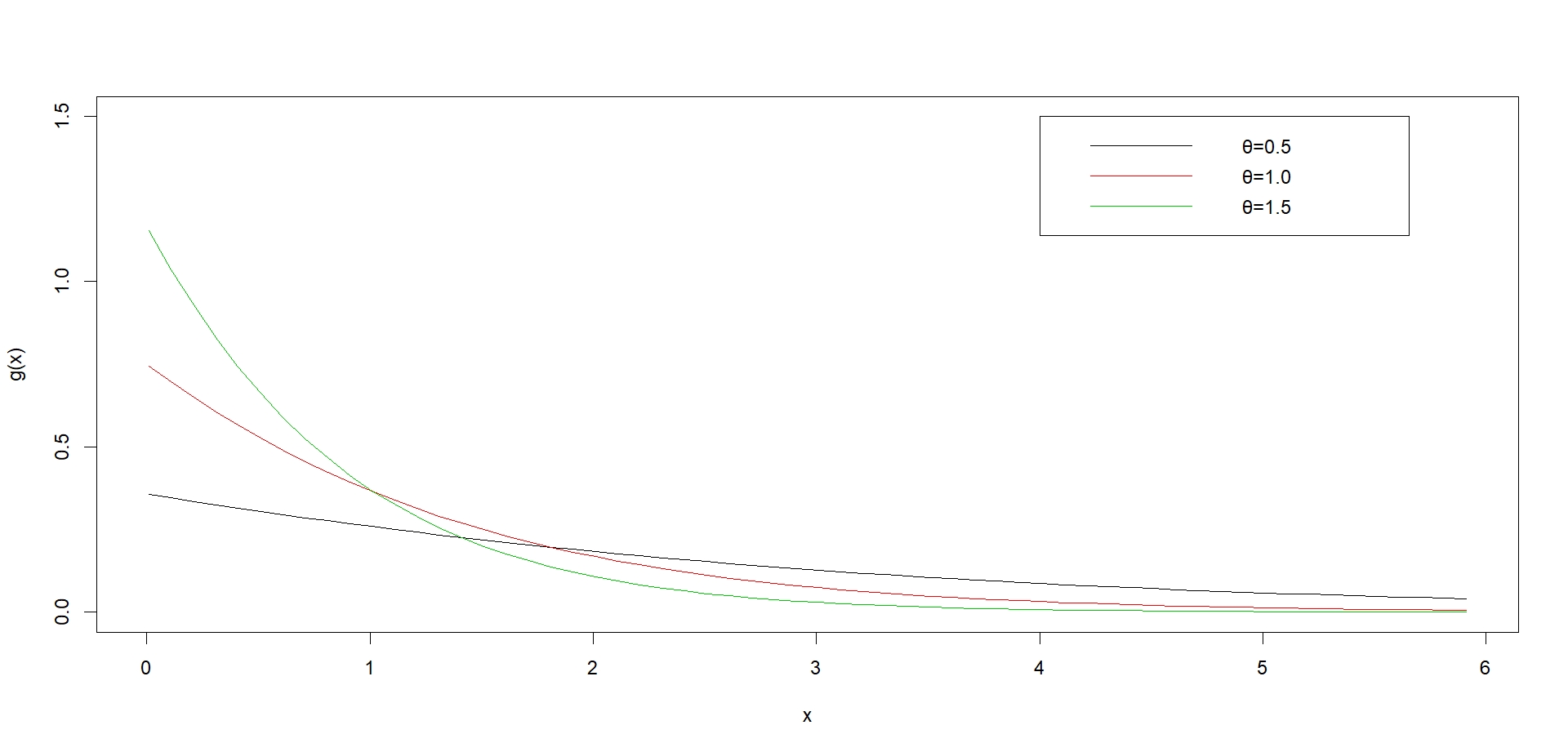}
	\caption{ Probability density function of $i$-Garima distribution}
	\label{fig:pdf}	
\end{figure}
\begin{figure}[H]
	\centering
	\includegraphics[width=0.75\linewidth]{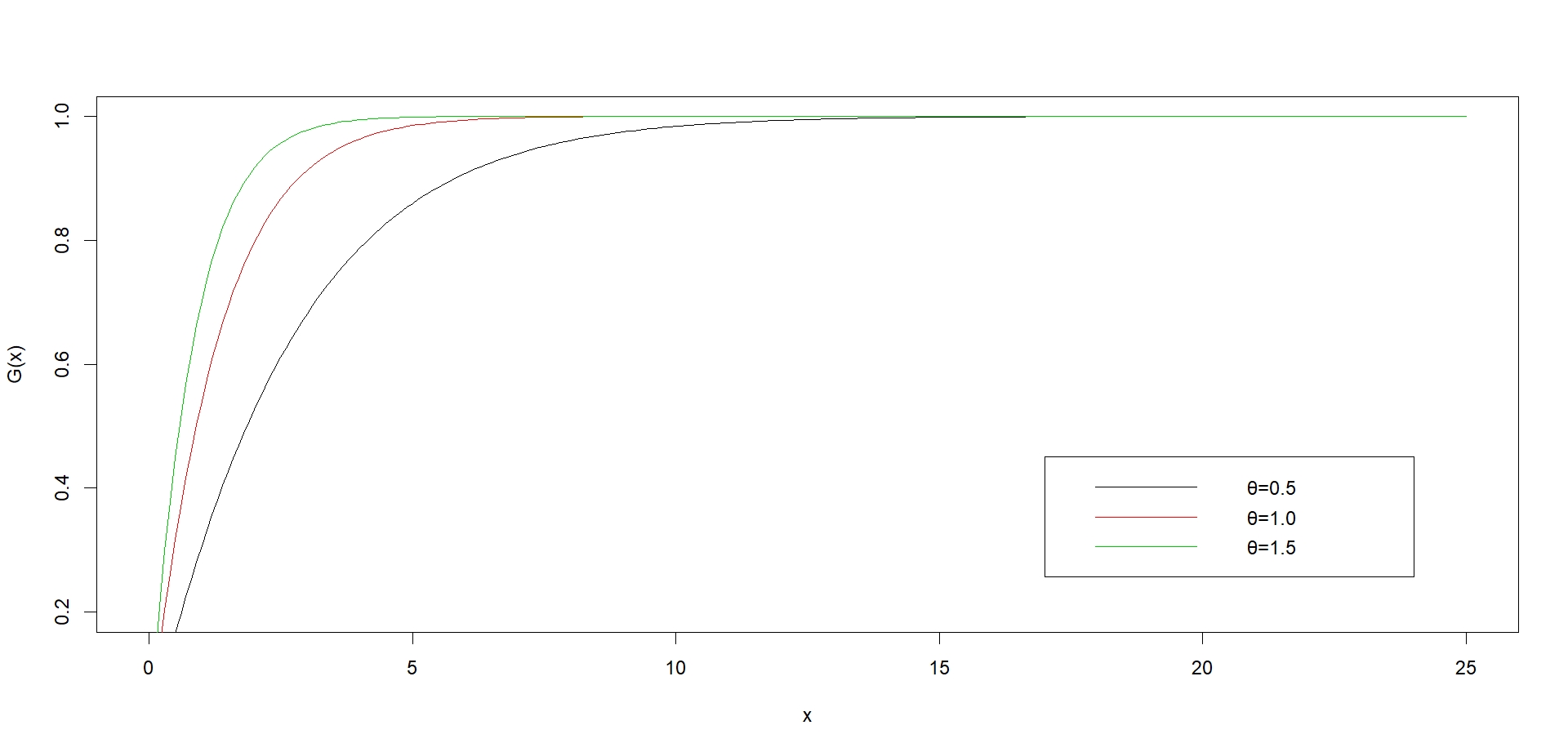}
	\caption{Cumulative distribution function of $i$-Garima distribution}
	\label{fig:cdf}	
\end{figure}
\noindent $i$-Garima distribution can be easily expressed as a mixture of Exponential($\theta$) and Gamma (2,$\theta$) with mixing proportion $\frac{\theta+2}{\theta+3}$. We have
\begin{align}
	f(x;\theta)=pg_1{(x)}+(1-p)g_{2}{(x)}
\end{align}
where $p=\frac{\theta+2}{\theta+3}, g_{1}{(x)}=\theta e^{-\theta x}$, and $g_{2}{(x)}=\theta^{2} x e^{-\theta x}$
\section{Properties}
The $r$-th order momemts about origin is given by  
\begin{equation}
	E(x^r)=\int\limits_0^{\infty}x^rg(x)dx=\frac{\theta}{\theta+3} \int\limits_0^{\infty}x^re^{-\theta x}\left(2+\theta+\theta x\right)dx\nonumber\\
\end{equation}
	Hence
	\begin{align}
		\mu'_r=\frac{r!}{\theta^r}\frac{\left(\theta+r+3\right)}{\left(\theta+3\right)}\qquad;r=1,2,3,...
	\end{align}
	Hence first four moments about origin is obtained as\\
	\begin{align}
		\mu'_1=\frac{1}{\theta}\frac{\left(\theta+4\right)}{\left(\theta+3\right)};\qquad \mu'_2=\frac{2}{\theta^2}\frac{\left(\theta+5\right)}{\left(\theta+3\right)};\qquad 
		\mu'_3=\frac{6}{\theta^3}\frac{\left(\theta+6\right)}{\left(\theta+3\right)};\qquad
		\mu'_4=\frac{24}{\theta^4}\frac{\left(\theta+7\right)}{\left(\theta+3\right)}\nonumber
	\end{align}
Using the above expression we get the four moments about mean, i.e. central moments of the proposed distribution is given by
\begin{align}
	\mu_1&=\frac{\theta+4}{\theta\left(\theta+3\right)};\qquad
	\mu_2=\frac{\theta^2+8\theta+14}{\theta^2\left(\theta+3\right)^2};\qquad\nonumber\\
	\mu_3&=\frac{2\left(\theta^3+12\theta^2+42\theta+46\right)}{\theta^3\left(\theta+3\right)^3};\qquad
	\mu_4=\frac{3\left(3\theta^4+48\theta^3+260\theta^2+592\theta+488\right)}{\theta^4\left(\theta+3\right)^4};\qquad\nonumber
\end{align}	
The coefficient of variation (CV), coefficient of skewness $\sqrt{\beta_1}$, coefficient of skewness $\beta_2$ and index of dispersion $\gamma$ of proposed distribution is obtained as
\begin{align}
	C.V&=\frac{\sigma}{\mu}=\frac{\sqrt{\theta^2+8\theta+14}}{\theta+4};\qquad
	\sqrt{\beta_1}=\frac{\mu_3}{\mu_2^{\frac{3}{2}}}=\frac{2\left(\theta^3+12\theta^2+42\theta+46\right)}{\left(\theta^2+8\theta+14\right)^{\frac{3}{2}}}\nonumber\\
	\beta_2&=\frac{\mu_4}{\mu_2^2}=\frac{3\left(3\theta^4+48\theta^3+260\theta^2+592\theta+488\right)}{\left(\theta^2+8\theta+14\right)^2};\qquad
	\gamma=\frac{\mu_2}{\mu_1}=\frac{\left(\theta^2+8\theta+14\right)}{\theta(\theta+3)(\theta+4)}\nonumber
\end{align}
\section{Generating functions}
The moment generating function $M_x(t)$, characteristic function $\Phi_x(t)$ and cumulant generating function $\kappa_x(t)$ of proposed distribution are given by
\begin{align}
	M_x(t)&=\left[1-\frac{(2+\theta)t}{(3+\theta)\theta}\right]\left(1-\frac{t}{\theta}\right)^{-2};{\frac{t}{\theta}}\\
	\Phi_x(t)&=\left[1-\frac{(2+\theta)it}{(3+\theta)\theta}\right]\left(1-\frac{it}{\theta}\right)^{-2}; \qquad i=\sqrt{-1}\\
	\kappa_x(t)&=\log\left(1-\frac{(2+\theta)it}{(3+\theta)\theta}\right)-2\log\left(1-\frac{it}{\theta}\right)
\end{align}
By series expansion of $\log(1-x)=-\sum\limits_{r=0}^{\infty}\frac{x^r}{r}$,we get
\begin{align}
	\kappa_x(t)&=-\sum\limits_{r=0}^{\infty}\left(\frac{(2+\theta)}{(3+\theta)\theta}\right)^r\frac{(it)^{r}}{r}+2\sum\limits_{r=0}^{\infty}\frac{\left(\frac{it}{\theta}\right)^r}{r}\nonumber\\
	&=2\sum\limits_{r=0}^{\infty}\frac{(r-1)!}{\theta^r}\frac{(it)^r}{r!}-\sum\limits_{r=0}^{\infty}(r-1)!\left[\frac{\theta+2}{\theta(\theta+3)}\right]^{r}\frac{(it)^r}{r!}\nonumber
\end{align}
Hence $r$-th cumulant of $i$-Garima distribution is given by\\
$\kappa_r$=coefficient of $\frac{(it)^r}{r!}$ in $\kappa_x(t)$

\begin{align}
=2\frac{(r-1)!}{\theta^r}-\frac{(r-1)!(\theta+2)^r}{\left[\theta(\theta+3)\right]^r};\qquad r=1,2,3,...\nonumber
\end{align}
From this we get the four moments which are same as obtained earlier by equation (1)
\begin{align}
	\mu_1&=\kappa_1=\frac{\theta+4}{\theta(\theta+3)};\qquad 
	\mu_2=\kappa_2=\frac{\theta^2+8\theta+14}{\theta^2\left(\theta+3\right)^2};\qquad\nonumber\\
	\mu_3&=\kappa_3=\frac{2\left(\theta^3+12\theta^2+42\theta+46\right)}{\theta^3\left(\theta+3\right)^3};\qquad
	\mu_4=\kappa_4+3\kappa_2^2=\frac{3\left(3\theta^4+48\theta^3+260\theta^2+592\theta+488\right)}{\theta^4\left(\theta+3\right)^4};\qquad\nonumber
\end{align}
\section{Hazard rate function and mean residual life function}
Let $X$ be a random variable with pdf $g(x)$ and cdf $G(x)$. The hazard function and the mean residual life function (MRLF) is given as respectively
\begin{align}
	h(x)&=\lim\limits_{\Delta x\to\infty}\frac{P(X<x+\Delta x|X>x)}{\Delta x}=\frac{g(x;\theta)}{1-G(x; \theta)}\\
	m(x)&=E\left[X-x|X>x\right]=\frac{1}{1-G(x;\theta)}\int\limits_x^{\infty}\left[1-G(t;\theta)\right]dt
\end{align}
Now we put the value of pdf and cdf of $i$-Garima distribution in above expression we get the hazard rate function $h(x)$ and MRLF $m(x)$ of $i$-Garima distribution as
\begin{align}
	h(x)&=\frac{\theta(2+\theta+\theta x)}{\left(3+\theta+\theta x\right)}\\
	m(x)&=\frac{\left(4+\theta+\theta x\right)}{\theta\left(3+\theta+\theta x\right)}\\
	\lim\limits_{x\to 0}h(x)&=\lim\limits_{x\to 0}\theta\left[1-\frac{1}{(3+\theta+\theta x)}\right]=\theta\left[1-\frac{1}{(3+\theta)}\right]>0; \theta\in\mathbb{R^{+}}\nonumber
\end{align}
As $x\to\infty$ we get
\begin{align}
	\lim\limits_{x\to\infty}h(x)=\lim\limits_{x\to\infty}\theta\left[1-\frac{1}{(3+\theta+\theta x)}\right]=\theta>0;\theta\in\mathbb{R^{+}}\nonumber
\end{align}
Hence, $h(x)>0,  x>0 ,\theta>0$	

\begin{figure}[H]
	\centering
	\includegraphics[width=0.75\linewidth]{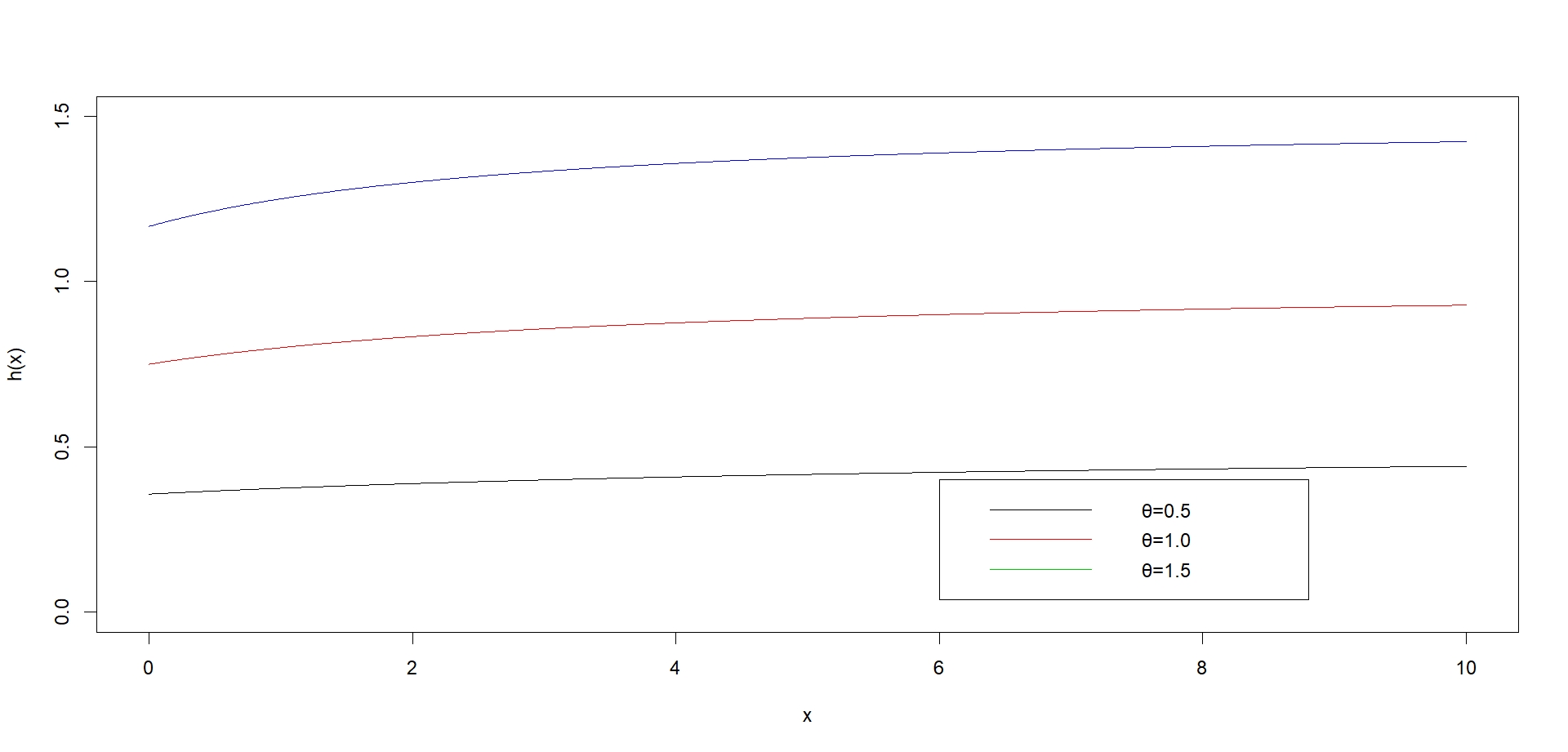}
	\caption{ Hazard function of $i$-Garima distribution}
	\label{fig:hazard}	
\end{figure}

We can see that $h(x)$ is an increasing function for the value $x>0$ and $\theta>0$ and $m(0)=\frac{\theta+4}{\theta(\theta+3)}$ which is mean of the distribution and also $m(x)$ is decreasing function for $x>0$ and $\theta>0$ .
\begin{figure}[H]
	\centering
	\includegraphics[width=0.75\linewidth]{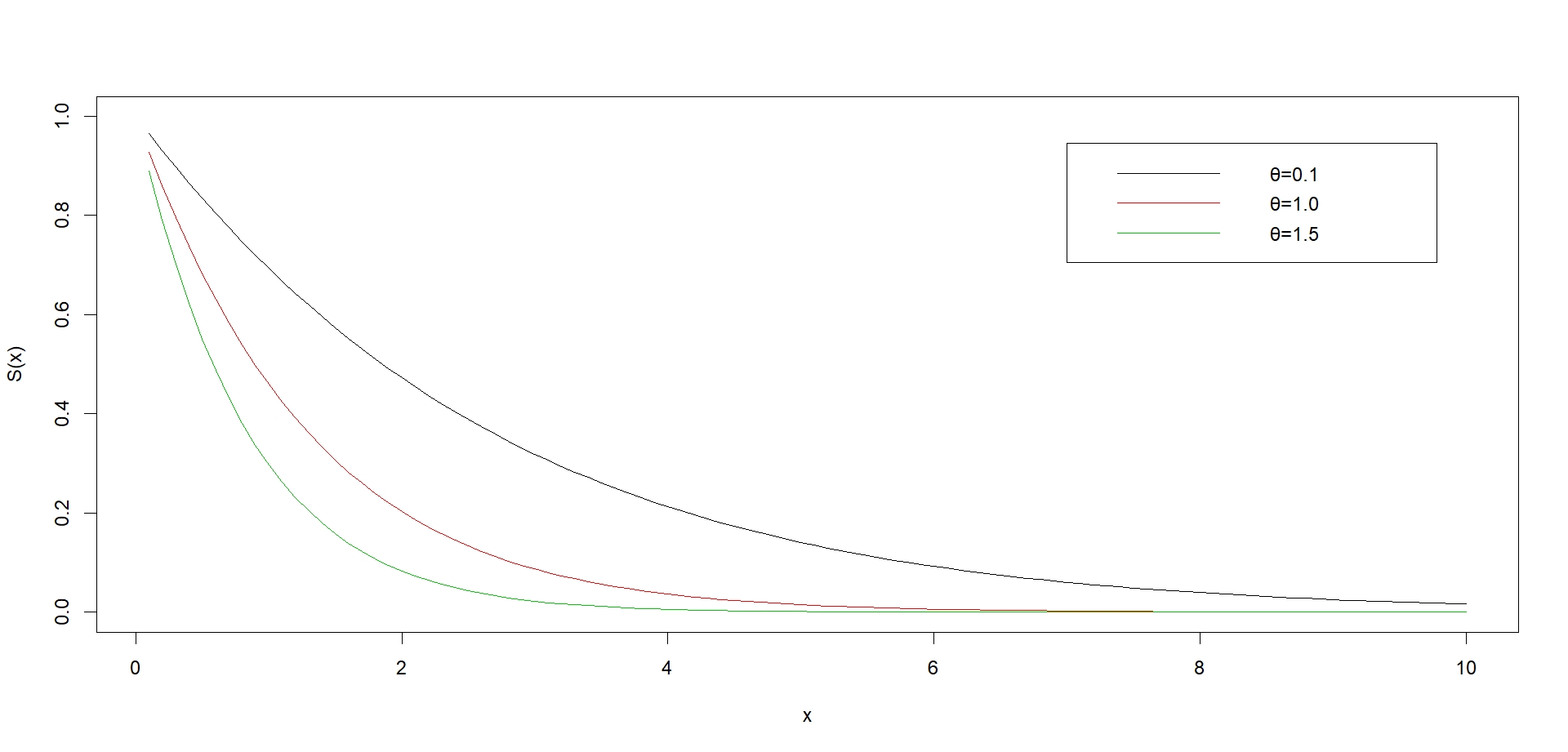}
	\caption{ Survival function of $i$-Garima distribution}
	\label{fig:survival}	
\end{figure}
\section{Quantile Function}
\begin{theorem}
	If $X\sim i-Garima(\theta)$, then Quantile function of $X$ is defined as
	\begin{align}
		Q(p)=-1-\frac{3}{\theta}-\frac{1}{\theta}W_{-1}\left( -(1-p)(\theta+3)e^{-(-\theta+3)}\right)\nonumber
	\end{align}
	where $p\in (0,1)$ and $W_{-1}$ is the negetive branch of the Lambert W function.
\end{theorem}
\begin{proof}
	Let,
	\begin{align}
		Q(p)=F^{-1}(p), p\in (0,1)\nonumber
	\end{align}
	The quantile function, say $q(p)$, defined by $G(Q(p))=p$ is the root of the equation
	\begin{align}
		1-\left( 1+\frac{\theta Q(p)}{\theta+3}\right) e^{-\theta Q(p)}=p\nonumber\\
		\left[ 3+\theta+\theta Q(p)\right] e^{-\theta Q(p)}=(1-p)(\theta+3)\nonumber
	\end{align}
	Multiplying both side by $-e^{-(\theta+3)}$ we get
	\begin{align}
		-\left[ 3+\theta+\theta Q(p)\right] e^{-\left(3+\theta+\theta Q(p)\right)}=-(1-p)(\theta+3)e^{-(3+\theta)}\nonumber
	\end{align}
	Now $\left(3+\theta+\theta Q(p)\right)>1 \forall{\theta}>0, Q(p)>0$
	By applying W-function defined as the solution of the equation $w(z)e^{W(z)}=z$, the above equation can be written as 
	\begin{align}
		W_{-1}\left( -(1-p)(\theta+3)e^{-(-\theta+3)}\right)=-\left(3+\theta+\theta Q(p)\right)\nonumber	
	\end{align}
	where and $W_{-1}(.)$ is the negetive branch of the Lambert W function and we get the required result
	\begin{align}
		Q(p)=-1-\frac{3}{\theta}-\frac{1}{\theta}W_{-1}\left( -(1-p)(\theta+3)e^{-(-\theta+3)}\right)
	\end{align}
\end{proof}
\section{Stochastic Orderings}
Stochastic ordering of a continuous random variable is an important tool to judging their comparative behaviour. A random variable $X$ is said to be smaller than a random variable $Y$.\\
(i) Stochastic order $X\leq_{st}Y$ if $F_{X}{(x)}\geq F_{Y}{(x)}$ for all x.\\
(ii) Hazard rate order $X\leq_{hr}Y$ if $h_{X}{(x)}\geq h_{Y}{(x)}$ for all x.\\
(iii) Mean residual life order $X\leq_{mrl}Y$ if $m_{X}{(x)}\geq m_{Y}{(x)}$ for all x.\\
(iv) Likelihood ratio order $X\leq_{lr}Y$ if $\frac{f_{X}{(x)}}{f_{Y}{(x)}}$ decreases in x.\\
The following results by \cite{shaked2007stochastic} are well known for introducing stochastic ordering of distributions
\begin{align}
	X\leq_{lr}Y\implies X&\leq_{hr}Y\implies X\leq_{mrl}Y\nonumber\\
	&\Downarrow\nonumber\\ X&\leq_{st}Y\nonumber
\end{align}
With the help of following theorem we claim that $i$-Garima distribution is ordered with respect to strongest likelihood ratio ordering

\begin{theorem}
	Let $X\sim i-Garima(\theta_1)$ distribution and $Y\sim i-Garima(\theta_2)$ distribution. If $\theta_1>\theta_2 $ then $X\leq_{lr}Y$ and therefore $X\leq_{hr}Y$, $X\leq_{mrl}Y$ and $X\leq_{st}Y$.
\end{theorem}
\begin{proof}
	We have 
	\begin{align}
		\frac{f_{X}(x)}{f_{Y}(x)}=\frac{\theta_1(\theta_2+3)}{\theta_2(\theta_1+3)}\left(\frac{2+\theta_1+\theta_1 x}{2+\theta_2+\theta_2 x}\right)e^{-\left(\theta_1-\theta_2\right)x};\qquad x>0\nonumber
	\end{align}
	Now taking log both side we get
	\begin{align}
		\log\frac{f_{X}(x)}{f_{Y}(x)}=\log\left[\frac{\theta_1(\theta_2+3)}{\theta_2(\theta_1+3)}\right]{-\left(\theta_1-\theta_2\right)x}\nonumber
	\end{align}
	By differentiating both side we get
	\begin{align}
		\frac{d}{dx}\log\frac{f_{X}(x)}{f_{Y}(x)}=\frac{\theta_1-\theta_2}{(2+\theta_1+\theta_1 x)(2+\theta_2+\theta_2 x)}{-\left(\theta_1-\theta_2\right)}\nonumber
	\end{align}
	Thus for $\theta_1>\theta_2 ,\frac{d}{dx}\log\frac{f_{X}(x)}{f_{Y}(x)}<0. $This means that $X\leq_{lr}Y$ and hence $X\leq_{hr}Y$, $X\leq_{mrl}Y$ and $X\leq_{st}Y$.
\end{proof}
\section{Order statistics}
Let $X_{1},X_{2},...,X_{m}$ be a random sample of size $m$ from $i$-Garima distribution and also let $X_{(1)},X_{(2)},...,X_{(m)}$ be the corresponding order statistics.The pdf and cdf of $r^{th}$ order statistics say $Y=X_{(r)}$ are given by
\begin{align}
	f_{(r:m)}(y)=\frac{m!}{(r-1)!(m-r)!}F^{r-1}(y)\left[1-F(y)\right]^{m-r}f(y)\nonumber\\
	=\frac{m!}{(r-1)!(m-r)!}\sum_{l=0}^{m-r}{{m-r}\choose_{l}}(-1)^lF^{r+l-1}(y)f(y)
\end{align}
and
\begin{align}
	F_{(r:m)}(y)=\sum_{j=r}^{m}{{m}\choose_{j}}F^{j}(y)\left[1-F(y)\right]^{m-j}\nonumber\\
	=\sum_{j=r}^{m}\sum_{l=0}^{m-j}{{m}\choose_{j}}{{m-j}\choose_{l}}(-1)^lF^{j+l}(y)
\end{align}
respectively, for $r=1(1)m$\\
Based on equation (16) and (17) the pdf and cdf of $r$-th order statistics of $i$-Garima distribution is given by
\begin{align}
	f_{(r:m)}(y)=\frac{m!\theta(3+\theta+\theta x)e^{-\theta x}}{(\theta+3)(r-1)!(m-r)!}\sum_{l=0}^{m-r}{{m-r}\choose{l}}\left[1-\frac{\theta x +(\theta+3)}{(\theta+3)}e^{-\theta x}\right]^{r+l-1}
\end{align}
and
\begin{align}
	F_{(r:m)}(y)=\sum_{j=r}^{m}\sum_{l=0}^{m-j}{{m}\choose_{j}}{{m-j}\choose_{l}}\left[1-\frac{\theta x +(\theta+3)}{(\theta+3)}e^{-\theta x}\right]^{j+l}
\end{align}
\section{Bonferroni and Lorenz curves}
Let the random variable $X$ is non-negative with a continuous and twice differentiable cumulative function.The \cite{bonferroni1936statistical} curve of the random variable $X$ is defined as
\begin{align}
	B(p)=\frac{1}{p\mu}\int\limits_{0}^{q}xg(x)dx=\frac{1}{p\mu}\left[\int\limits_{0}^{\infty}xg(x)dx-\int\limits_{q}^{\infty}xg(x)dx\right]=\frac{1}{p\mu}\left[\mu-\int\limits_{q}^{\infty}xg(x)dx\right]
\end{align}
and the \cite{lorenz1905methods} curve is defined by
\begin{align}
	L(p)=\frac{1}{\mu}\int\limits_{0}^{q}xg(x)dx=\frac{1}{\mu}\left[\int\limits_{0}^{\infty}xg(x)dx-\int\limits_{q}^{\infty}xg(x)dx\right]=\frac{1}{\mu}\left[\mu-\int\limits_{q}^{\infty}xg(x)dx\right]
\end{align}
where $q=G^{-1}(p)$ and $\mu=E(X)$, $p\in (0,1] $\\
The Gini index is given by
\begin{align}
	G=1-\frac{1}{\mu}\int\limits_{0}^{\infty}\left(1-G(x)\right)^2dx=\frac{1}{\mu}\int\limits_{0}^{\infty}G(x)\left(1-G(x)\right)dx
\end{align}
The Bonferroni, Lorenz curve and Gini index have application not only in economics to study income and poverty, but also in other fields like reliability,population studies,insurence, medicine. Using the equation (20) and (21) in (22) we get the Bonferroni curve, Lorenz curve and the Gini index as
\begin{align}
	B(p)=\frac{1}{p}\left[1-\frac{\{\theta^2 q^2+(\theta^2+4\theta)q+(\theta+4)\}e^{-\theta q}}{\theta+4}\right]\\
		L(p)=1-\frac{\{\theta^2 q^2+(\theta^2+4\theta)q+(\theta+4)\}e^{-\theta q}}{\theta+4}
\end{align}
and
\begin{align}
	G=\frac{2\theta^2+16\theta+29}{4(\theta+3)(\theta+4)}
\end{align}
\section{Entropies}
\subsection{Renyi Entropy}
An entropy is a measure of variation of the uncertainty,  \cite{renyi1961measures} gave an expression of the Entropy function defined by. If $X$ is a continuous random variable having probability density function $g(.)$, then the Renyi Entropy is defined as
\begin{align}
	e(\eta)=\frac{1}{1-\eta}\log\left[\int\limits_0^{\infty}g^{\eta}(x)dx\right]	
\end{align}
where $\eta>0$ and $\eta=0$.\\
The Renyi Entropy for the $i$-Garima distribution is defined as
\begin{align}
	e(\eta)&=\frac{1}{1-\eta}\log\left[\int\limits_0^{\infty}\left(\frac{\theta}{\theta+3}\right)^{\eta}\left(2+\theta+\theta x\right)^{\eta}e^{-\eta\theta x} dx\right]\nonumber\\
	&=\frac{1}{1-\eta}\log\left[\int\limits_0^{\infty}\frac{\theta^{\eta}(\theta+2)^{\eta}}{(\theta+3)^{\eta}}\left(1+\frac{\theta x}{\theta+2}\right)^{\eta}e^{-\eta\theta x}dx\right]\nonumber\\
	&=\frac{1}{1-\eta}\log\left[\int\limits_0^{\infty}\frac{\theta^{\eta}(\theta+2)^{\eta}}{(\theta+3)^{\eta}} \sum_{j=0}^{\infty}{{\eta}\choose_{j}}\left(\frac{\theta x}{\theta+2}\right)^{j} e^{-\eta\theta x}dx\right]\nonumber\\
	&=\frac{1}{1-\eta}\log\left[\sum_{j=0}^{\infty}{{\eta}\choose_{j}}\frac{\theta^{\eta+j}(\theta+2)^{\eta-j}}{(\theta+3)^{\eta}}\int\limits_0^{\infty}x^{j}e^{-\eta\theta x}dx\right]\nonumber\\
	&=\frac{1}{1-\eta}\log\left[\sum_{j=0}^{\infty}{{\eta}\choose_{j}}\frac{\theta^{\eta+j}(\theta+2)^{\eta-j}}{(\theta+3)^{\eta}}\int\limits_0^{\infty}x^{j}e^{-\eta\theta x}dx\right]\nonumber\\
	&=\frac{1}{1-\eta}\log\left[\sum_{j=0}^{\infty}{{\eta}\choose_{j}}\frac{\theta^{\eta+j}(\theta+2)^{\eta-j}}{(\theta+3)^{\eta}}\frac{\Gamma{(j+1)}}{(\eta\theta)^{j+1}}\right]\nonumber\\
	&=\frac{1}{1-\eta}\log\left[\sum_{j=0}^{\infty}{{\eta}\choose_{j}}\frac{\theta^{\eta-1}(\theta+2)^{\eta-j}}{(\theta+3)^{\eta}}\frac{\Gamma{(j+1)}}{(\eta)^{j+1}}\right]
\end{align}
\subsection{Shannon Entropy}
The Shannon Entropy by \cite{shannon1951prediction} of $i$-Garima distribution is given as
\begin{align}
	E(-\log x)&=-\int\limits_0^{\infty}\log(f(x))f(x)dx\nonumber\\
	&=-\log\left(\frac{\theta}{\theta+3}\right)\int\limits_0^{\infty}f(x)dx-\int\limits_0^{\infty}\log\left(2+\theta+\theta x\right)f(x)dx+\int\limits_0^{\infty}\theta xf(x)dx\nonumber\\
	&=-\log\left(\frac{\theta}{\theta+3}\right)-\log(\theta+2)-\int\limits_0^{\infty}\left(1+\frac{\theta x}{\theta+2}\right)f(x)dx+\theta E(x)\nonumber\\
	&=-\log\left(\frac{\theta(\theta+2)}{\theta+3}\right)+\left(\frac{\theta+4}{\theta+3}\right)-\frac{\theta}{\theta+3}\int\limits_0^{\infty}\sum_{k=1}^{\infty}\frac{(-1)^{k+1}}{k}\left(\frac{\theta x}{\theta+2}\right)^{k}(2+\theta+
	\theta x) e^{-\theta x}dx\nonumber\\
	&=-\log\left(\frac{\theta(\theta+2)}{\theta+3}\right)+\left(\frac{\theta+4}{\theta+3}\right)-\frac{\theta}{\theta+3}\int\limits_0^{\infty}\sum_{k=1}^{\infty}\frac{(-1)^{k+1}}{k}\left(\frac{\theta x}{\theta+2}\right)^{k}(2+\theta+
	\theta x) e^{-\theta x}dx\nonumber\\
	&=-\log\left(\frac{\theta(\theta+2)}{\theta+3}\right)+\left(\frac{\theta+4}{\theta+3}\right)-\frac{\theta}{\theta+3}\sum_{k=1}^{\infty}\frac{(-1)^{k+1}}{k}\left(\frac{\theta}{\theta+2}\right)^{k}\int\limits_0^{\infty}x^{k}(2+\theta+\theta x) e^{-\theta x}dx\nonumber\\
	&=\left(\frac{\theta+4}{\theta+3}\right)-\log\left(\frac{\theta(\theta+2)}{\theta+3}\right)-\frac{1}{\theta+3}\sum_{k=1}^{\infty}\frac{(-1)^{k+1}}{k}\frac{k!(\theta+k+3)}{(\theta+2)^k}
\end{align}
\section{Stress-strength reliability}
Stress strength-model describes the life of a component having a random strength $X$ and subjected to random stress $Y$. The component function satisfactorily for $X>Y$ and fails when $Y>X$.Therefore $R=P(Y<X)$ is a measure of components reliability and in statistical literature it is known as stress-strength parameter.It has many applications like in engineering concepts such as structures, deterioration of rocket motors, static fatigue of ceramic components, fatigue failure of aircraft structures and aging of cocrete pressure vessels.\\
Let $X$ and $Y$ be independently distributed, with $X\sim i-Garima(\theta_1)$ and $Y\sim i-Garima(\theta_2)$. The CDF $F_1$ of $X$ and pdf $f_2$ of $Y$ are obtained from equation (3) and (2),respectively. Then stress-strength reliability $R$ is obtained as
\begin{align}
	R&=P(Y<X)=\int\limits_0^{\infty}P(Y<X|X=x)f_{x}(X)dx=\int\limits_0^{\infty}f(x;\theta_1)F(x;\theta_2)dx\nonumber\\
	&=1-\frac{\theta_1\left[(\theta_1\theta_2+3\theta_1+2\theta_2+6)(\theta_1+\theta_2)^2+(2\theta_1\theta_2+3\theta_1+2\theta_2)(\theta_1+\theta_2)+2\theta_1\theta_2\right]}{(\theta_1+3)(\theta_2+3)(\theta_1+\theta_2)^3}	
\end{align}
\section{Maximum likelihood estimates(MLE)}
Let $(x_1,x_2,...,x_n)$ be a random sample from $X\sim i-Garima(\theta)$. The likelihood function, $L$ is obtained as
\begin{align}
	L=\left(\frac{\theta}{\theta+3}\right)^n\prod_{i=1}^{n}(2+\theta+\theta x_i)e^{-\theta\sum\limits_{i=0}^{n}x_i}
\end{align}
Taking log both side we get
\begin{align}
	\log L=n\log\left(\frac{\theta}{\theta+3}\right)+\sum_{i=1}^{n}\log(2+\theta+\theta x_i)-\theta\sum\limits_{i=1}^{n}x_i
\end{align}
Now differentiate both side by $\theta$ we get 
\begin{align}
	\frac{d(\log L)}{d\theta}=\frac{3n}{\theta^2+3\theta}+\sum_{i=1}^{n}\frac{1+x_i}{2+\theta+\theta x_i}-n\bar{x}=0
\end{align}
Where $\bar{x}$ is the sample mean.\\
The maximum likelihood estimate, $\hat{\theta}$ of $\theta$ is the solution of the equation (32), since this is non-linear equation we solve this by numerical method.\\
\section{Goodness of Fit}
The goodness of fit of the $i$-Garima distribution has been explained with four real data sets, first data is vinyl chloride data obtained from clean upgradient monitoring wells in mg/l, provided by \cite{bhaumik2009testing}, second data set represents remission times (in months) of a random sample of 128 bladder cancer patients reported in \cite{lee2003statistical}, third data set is given by \cite{linhart1986model}, which represents the failure times of the air conditioning system of an airplane and last data set is given by  \cite{lawless2011statistical} concerning the data on time to breakdown of an insulating fluid between electrodes at a voltage of 34kV (minutes). We fitted the $i$-Garima distributions for the above data sets and compared
the results with the \cite{Garima2018Shanker}, \cite{shanker2015aradhana},  \cite{shanker2016sujatha}, \cite{shanker20155akash},  \cite{shanker2015shanker} and \cite{lindley1958fiducial} Distributions with there densities. In order to compare lifetime distributions, values of $-2lnL$, AIC (Akaike Information Criterion), BIC (Bayesian Information Criterion) and K-S Statistic (Kolmogorov-Smirnov Statistic) for the above data sets has been computed. The formulae for computing AIC, BIC, and K-S Statistics are as follows:
\begin{align}
	AIC=-2loglik+2k,&\qquad BIC=-2loglik+k\log n\nonumber\\
	D&=\sup\limits_{x}|F_{n}{(x)}-F_{0}{(x)}|\nonumber
\end{align}
where $k$= the number of parameters, $n$= the sample size, and the $F_{n}{(x)}$=empirical distribution function
From the table it is shows that the $i$-Garima distribution provides the best fit for the above data sets as it has lower -2LL,AIC, BIC, K-S values and higher p-values corresponding to K-S statistics than the other competitor models. It is obvious that $i$-Garima gives much closer fit than other compared distributions. Therefore, $i$-Garima can be considered as an important one-parameter lifetime distribution.
\begin{table}[H]
\centering
\caption{-2LL,AIC,BIC,K-S statistics for the data sets.}
\begin{tabular}{|c|c|c|c|c|c|c|c|}
	\hline
	\multirow{1}{*}{Data Sets} & \multicolumn{1}{c|}{Distrbutions} & %
	\multicolumn{1}{c|}{Estimate} & \multicolumn{1}{c|}{-2LL} & \multicolumn{1}{c|}{AIC} & \multicolumn{1}{c|}{BIC}  & \multicolumn{1}{c|}{K-S} & \multicolumn{1}{c|}{p-value} \\
	\cline{1-8}

    \multirow{7}{*}{1st} & \multicolumn{1}{c|}{$i$-Garima} & \multicolumn{1}{c|}{0.674} & \multicolumn{1}{c|}{111.18} & \multicolumn{1}{c|}{113.18} & \multicolumn{1}{c|}{114.71} & \multicolumn{1}{c|}{0.1039} & \multicolumn{1}{c|}{0.8567} \\
    \cline{2-8}
    & Garima & 0.723 & 111.50 & 113.50 & 115.03 &  0.1135 & 0.7731 \\
    \cline{2-8}
    & Aradhana & 1.133 & 116.06 & 118.06 & 119.59 & 0.1695 & 0.2826 \\
    \cline{2-8}
    & Sujatha & 1.146 & 115.54 & 117.54 & 119.07 & 0.1640 & 0.3196 \\
    \cline{2-8}
    & Akash & 1.166 & 115.15 & 117.15 & 118.68 & 0.1564 & 0.3762 \\
    \cline{2-8}
    & Shanker & 0.853 & 112.91 & 114.91 & 116.44 & 0.1308 & 0.6062 \\
    \cline{2-8}
    & Lindley & 0.199 & 112.61 & 114.61 & 116.13 &  0.1326 & 0.5881 \\
    \hline
	
	\multirow{7}{*}{2nd} & \multicolumn{1}{c|}{$i$-Garima} & \multicolumn{1}{c|}{0.143} & \multicolumn{1}{c|}{825.57} & \multicolumn{1}{c|}{827.57} & \multicolumn{1}{c|}{830.42} & \multicolumn{1}{c|}{0.0768} & \multicolumn{1}{c|}{0.4374} \\
	\cline{2-8}
	& Garima & 0.158 & 826.49 & 828.49 & 831.34 &  0.0873 & 0.2835 \\
	\cline{2-8}
	& Aradhana & 0.295 & 868.28 & 870.28 & 873.13 & 0.1713 & 0.0011 \\
	\cline{2-8}
	& Sujatha & 0.303 & 873.22 & 875.22 & 878.08 & 0.1792 & 0.0005 \\
	\cline{2-8}
	& Akash & 0.315 & 881.04 & 883.04 & 885.89 & 0.1904 & 0.0002 \\
	\cline{2-8}
	& Shanker & 0.214 & 841.68 & 843.68 & 846.53 & 0.1243 & 0.0382 \\
	\cline{2-8}
	& Lindley & 0.199 & 833.79 & 835.79 & 838.64 & 0.1114 & 0.0832 \\
	\hline
		
	\multirow{7}{*}{3rd} & \multicolumn{1}{c|}{$i$-Garima} & \multicolumn{1}{c|}{0.089} & \multicolumn{1}{c|}{140.54} & \multicolumn{1}{c|}{142.54} & \multicolumn{1}{c|}{143.49} & \multicolumn{1}{c|}{0.2770} & \multicolumn{1}{c|}{0.0883} \\
	\cline{2-8}
	& Garima & 0.098 & 142.10 & 144.10 & 145.04 & 0.3015 & 0.0499 \\
	\cline{2-8}
	& Aradhana & 0.196 & 167.37 & 169.37 & 170.32 & 0.4123 & 0.0019 \\
	\cline{2-8}
	& Sujatha & 0.200 & 169.22 & 171.22 & 172.16 & 0.4193 & 0.0015 \\
	\cline{2-8}
	& Akash & 0.206 & 171.95 & 173.95 & 174.89 & 0.4285 & 0.0011 \\
	\cline{2-8}
	& Shanker & 0.141 & 156.18 & 158.18 & 159.13 & 0.3534 & 0.0125 \\
	\cline{2-8}
	& Lindley & 0.131 & 151.08 & 153.08 & 154.03 & 0.3462 & 0.0154 \\
	\hline	
	
	\multirow{7}{*}{4th} & \multicolumn{1}{c|}{$i$-Garima} & \multicolumn{1}{c|}{0.022} & \multicolumn{1}{c|}{306.75} & \multicolumn{1}{c|}{308.75} & \multicolumn{1}{c|}{310.15} & \multicolumn{1}{c|}{0.2400} & \multicolumn{1}{c|}{0.0631} \\
	\cline{2-8}
	& Garima & 0.024 & 308.73 & 310.73 & 312.13 & 0.2657 & 0.0289 \\
	\cline{2-8}
	& Aradhana & 0.049 & 350.55 & 352.55 & 353.95 & 0.4154 & 0.0000 \\
	\cline{2-8}
	& Sujatha & 0.050 & 352.47 & 354.47 & 355.87 & 0.4182 & 0.0000 \\
	\cline{2-8}
	& Akash & 0.050 & 354.88 & 356.88 & 358.28 & 0.4213 & 0.0000 \\
	\cline{2-8}
	& Shanker & 0.033 & 325.74 & 327.74 & 329.15 & 0.3517 & 0.0012 \\
	\cline{2-8}
	& Lindley & 0.033 & 323.27 & 325.27 & 326.67 & 0.3452 & 0.0016 \\
	\hline
	
\end{tabular}
\end{table}

\section{Conclusions}

This paper suggests a new induced-transformation and a new transformed continuous one-parameter lifetime distribution named as $i$-Garima distribution. Moments about origin and moments about mean have been obtained. The nature of probability density function, cumulative distribution function, hazard rate function, stress-strength reliability and mean residual life function have been discussed. Bonferroni and Lorenz curves and the Gini index of the $i$-Garima as well as the stochastic ordering are presented. The maximum likelihood estimators of the model parameters are derived as well as distributions of order statistics are provided. The Renyi entropies, Shannon entopies, stochastic ordering are derived. A numerical example of four real lifetime data sets have been presented to show the application of $i$-Garima and the goodness of fit of $i$-Garima gives much better fit over Garima, Aradhana, Sujatha, Akash, Shanker and Lindley distributions.

\bibliographystyle{chicago}
\bibliography{myref}

\end{document}